\documentclass[11pt]{amsart}
\usepackage[a4paper,margin=1in]{geometry}
\usepackage{amsmath,amssymb,amsthm,mathtools}
\usepackage{bm}
\usepackage{enumitem}
\usepackage{microtype}
\usepackage[hidelinks]{hyperref}
\usepackage[nameinlink,capitalise,noabbrev]{cleveref}
\usepackage{verbatim}

\newcommand{\C}{\mathbb C}
\newcommand{\R}{\mathbb R}
\newcommand{\bell}{{ \ell}}
\newcommand{\T}{\mathbb T}
\newcommand{\CP}{\mathbb {CP}}
\newcommand{\cH}{\mathcal H}
\newcommand{\cM}{\mathcal M}
\newcommand{\tcM}{\widetilde{\mathcal M}}

\newcommand{\cQ}{\mathcal Q}
\newcommand{\ip}[2]{\langle #1,#2\rangle}
\newcommand{\norm}[1]{\lVert #1\rVert}

\newcommand{\Span}{\operatorname{span}}
\newcommand{\vol}{\operatorname{vol}}

\theoremstyle{plain}
\newtheorem{theorem}{Theorem}[section]
\newtheorem{lemma}[theorem]{Lemma}
\newtheorem{proposition}[theorem]{Proposition}
\newtheorem{corollary}[theorem]{Corollary}

\theoremstyle{definition}
\newtheorem{definition}[theorem]{Definition}

\theoremstyle{remark}
\newtheorem{remark}[theorem]{Remark}

\title[The Minimum Number for Almost-Everywhere Complex Phase Retrieval]{The Minimum Number of Measurements for Almost-Everywhere Complex Phase Retrieval}
\author{Zhiqiang Xu}
\thanks{Zhiqiang Xu is supported by National Natural Science Foundation of China (Grant No. 12471361, 12288201).}
\address{State Key Laboratory of Mathematical Sciences, Academy of Mathematics and Systems Science, Chinese Academy of Sciences, Beijing 100190, China;    School of Mathematical Sciences, University of Chinese Academy of Sciences, Beijing 100049, China. }
\email{xuzq@lsec.cc.ac.cn}
\date{}

\begin{document}

\maketitle

\begin{abstract}
Let \(d\geq2\) and let
\(\bm{f}_1,\ldots,\bm{f}_m\in\mathbb C^d\).
We prove that if \(m\leq 2d-1\), then the intensity measurement map
\[
    \bm{x}\longmapsto
    \bigl(
        |\langle \bm{x},\bm{f}_1\rangle|^2,
        \ldots,
        |\langle \bm{x},\bm{f}_m\rangle|^2
    \bigr)
\]
fails to recover almost every signal in \(\mathbb C^d\) uniquely up to a
global phase factor. 
Combined with the known generic sufficiency of \(2d\) measurements, our
result establishes that the minimum number of measurements required for
almost-everywhere phase retrieval in \(\mathbb C^d\) is exactly \(2d\).
This resolves an open problem in phase retrieval by
determining the exact measurement threshold for almost-everywhere phase
retrieval in $\C^d$.
\end{abstract}

\section{Introduction}

\subsection{Problem setup}

Finite-dimensional phase retrieval concerns the reconstruction of a vector, up to an unavoidable global phase factor, from the magnitudes of its linear measurements.
Beyond its intrinsic mathematical interest, this problem plays a central role in a wide range of scientific and engineering applications.
To formulate the problem studied in this paper, we first introduce the necessary notation and definitions.
Set
$
    \T:=\{\omega\in\C:|\omega|=1\}.
$
Define global-phase equivalence on $\C^d$ by
$   \bm{x}\sim\bm{y}$
if there exists $\omega\in\T$ such that
$    \bm{y}=\omega\bm{x}\ \text{for some }\omega\in\T$.
We use $[\bm{x}]$ to denote the equivalent class containing $\bm{x}$, i.e, $[{\bm x}]:=\{{\bm y}\in \C^d: {\bm y}\sim {\bm x}\}$.
We set
$
    \cQ_d:=\C^d/\T.
$
For $F=(\bm{f}_j)_{j=1}^m\subset\C^d$, define the intensity map
\begin{equation}\label{eq:intensity-map}
    \cM_F:\cQ_d\longrightarrow\R^m,
    \qquad
    \cM_F([\bm{x}])
    :=\bigl(|\ip{\bm{x}}{\bm{f}_j}|^2\bigr)_{j=1}^m.
\end{equation}
This is well defined because multiplication of $\bm{x}$ by a unimodular scalar
does not change any coordinate.

\begin{definition}\label{def:prae}
The family $F=(\bm{f}_j)_{j=1}^m\subset\C^d$ is \emph{phase retrievable almost everywhere} in $\C^d$ (PR-ae) if
\[
    \cM_F^{-1}\bigl(\cM_F([\bm{x}])\bigr)=\{[\bm{x}]\}
\]
for almost every $[\bm{x}]\in\cQ_d$. Here,  
 \[
 \cM_F^{-1}\bigl(\cM_F([\bm{x}])\bigr):=\{[{\bm y}]\in \cQ_d: \cM_F([{\bm y}])=\cM_F([{\bm x}])\}.
\]
\end{definition}

In this paper, we focus on the following question:

\begin{center}
{\bf Question:} 
 What is the minimum number of measurements \(m\) for which there exists a family of vectors
 $F=(\bm{f}_j)_{j=1}^m\subset\C^d$ which is \emph{phase retrievable almost everywhere} in $\C^d$?
\end{center}

 More precisely, we ask
whether a specially structured family of $m=2d-1$ vectors is phase retrievable almost everywhere in $\C^d$, even though a generic family of
that size is known to fail (see \cite[Theorem 5.3]{HuangRongWangXu2021}).  In this paper, we prove that this is impossible for every
family, and hence determine the minimal measurement number for PR-ae in $\C^d$ is $m=2d$.

\subsection{Related work}

In \cite{BalanCasazzaEdidin2006},  Balan, Casazza, and Edidin
 formulated finite-dimensional phase retrieval in
terms of frames and established generic reconstruction results.  In
particular, their Theorem~3.4 shows that, for a generic family of at least
$2d$ complex vectors, the set of uniquely recoverable signals has dense
interior.    Fickus, Mixon,
Nelson, and Wang \cite{FickusMixonNelsonWang2014} developed corresponding
almost-injectivity methods in the real setting; see also
\cite{Mixon2015} for a broader discussion of phase transitions in phase
retrieval.
In particular, for the real case, it is known that the minimum  number of
measurements required  is \(m= d+1\) \cite{BalanCasazzaEdidin2006,FickusMixonNelsonWang2014}, and
the paper \cite{HuangRongWangXu2021} established a characterization of the
measurement vectors \(\{\bm f_j\}_{j=1}^m\subset\mathbb R^d\) that possess  PR-ae in $\R^d$.


Note that the phase retrieval measurement can be written as a
quadratic measurement
\(
    |\langle \bm{x},\bm{f}_j\rangle|^2
    =
    \bm{x}^*\bm{f}_j\bm{f}_j^*\bm{x}.
\)
Thus,  phase retrieval corresponds to the special case of quadratic
measurements in which the measurement matrices are rank-one positive
semidefinite matrices. Generalized phase retrieval extends this framework by
allowing the measurement matrices to be arbitrary Hermitian matrices rather
than restricted to the form \(\bm{f}_j\bm{f}_j^*\).
Wang and Xu~\cite{WangXu2019} developed a general theory for this
framework, relating the measurement complexity to low-rank matrix
recovery and topological constraints.
In \cite{HuangRongWangXu2021}, the authors subsequently introduced and analyzed
almost-everywhere generalized phase retrieval   for quadratic
measurements 
\[
    \bm{x}^*A_j\bm{x},\qquad j=1,\ldots,m,
\]
where \(A_1,\ldots,A_m\in\mathbb C^{d\times d}\) are Hermitian matrices.
 In the complex generalized setting, their
Theorem~1.4 constructs $2d-1$ Hermitian matrices with generalized PR-ae.  By contrast, in
the standard rank-one setting $A_j=\bm{f}_j\bm{f}_j^*$, their Theorem~5.3
proves that a \emph{generic} family of $2d-1$ vectors is not PR-ae.  Their
Theorem~1.2 also gives the generic sufficiency of $2d$ complex rank-one
measurements.

There is also a closely related formulation in quantum tomography. 
A POVM measurement of a pure quantum state corresponds to a positive
semidefinite generalized phase retrieval map on \(\mathbb{C}^d\).
Flammia, Silberfarb, and Caves~\cite{FlammiaSilberfarbCaves2005}, and
Finkelstein~\cite{Finkelstein2004}, constructed \(2d\)-outcome measurements
that achieve almost-everywhere recovery of generic pure states, the latter
in the rank-one setting. Heinosaari, Mazzarella, and Wolf~\cite{HeinosaariMazzarellaWolf2013}
further investigated topological obstructions to such generalized phase
retrieval problems.

The distinction between generic impossibility and universal impossibility is worth noting in this problem.
 A generic argument in \cite[Theorem 5.3]{HuangRongWangXu2021} only shows that a typical family
of $2d-1$  vectors fails to have the PR-ae property; it does not exclude
the possibility that a specially structured family of measurements may still
achieve almost-everywhere phase retrieval.
 Moreover, the discussion preceding
Lemma~5.2 of \cite{HuangRongWangXu2021}, in agreement with
\cite{Mixon2015}, points out that the previously claimed impossibility of
all \(2d-1\)-measurement systems in \cite{FlammiaSilberfarbCaves2005}  had not been established by a generally
accepted rigorous argument.
 Consequently,  to the best of our knowledge, the question of whether there
exists a family of \(2d-1\) vectors  with the PR-ae property remained an
open problem. The purpose of the present paper is to resolve this question
 by proving that no family of \(2d-1\) 
measurement vectors can have the PR-ae property.
\subsection{Our contribution}

We first state our main results as follows:

\begin{theorem}[Main theorem]\label{thm:main}
Let $d\geq2$ and let
\[
    F=(\bm{f}_1,\ldots,\bm{f}_{m})\subset\C^d,
\]
where $m\leq 2d-1$.
Then
\[
    \mathcal U_F
    :=\left\{[\bm{x}]\in\cQ_d:
       \cM_F^{-1}\bigl(\cM_F([\bm{x}])\bigr)=\{[\bm{x}]\}
      \right\}
\]
has measure zero in $\cQ_d$.  In particular, $F$ is not PR-ae.
\end{theorem}

Theorem~1.2 of \cite{HuangRongWangXu2021} states that, for $m = 2d$, a
generic family of $m$ complex vectors has the PR-ae property.  Historically, Theorem~3.4 of Balan, Casazza, and
Edidin \cite{BalanCasazzaEdidin2006} provided an earlier generic
dense-interior uniqueness result in the complex setting.
Combining Theorem \ref{thm:main} with the known generic upper bound gives
the exact threshold.

\begin{corollary}\label{cor:exact-minimum}
For every $d\geq 2$, the minimum number of vectors in a family
$F=(\bm{f}_1,\ldots,\bm{f}_m)\subset \C^d$
having the PR-ae property is exactly $2d$.
\end{corollary}

\begin{remark}

The same argument extends to quadratic measurements induced by positive
semidefinite matrices. In particular, we obtain the following consequence.
Let $d\geq2$ and let $A_1,\ldots,A_{2d-1}\succeq0$ be Hermitian matrices on
$\C^d$.  Then the quadratic measurement map
\[
    [\bm{x}]\longmapsto
    (\bm{x}^*A_1\bm{x},\ldots,\bm{x}^*A_{2d-1}\bm{x})
\]
is not PR-ae. 
There is no contradiction with Theorem~1.4 of
\cite{HuangRongWangXu2021}, which  constructs $2d-1$ generalized
measurements with PR-ae.   
The distinction lies in the fact that the generalized measurements considered in \cite{HuangRongWangXu2021} allow arbitrary Hermitian
operators and are therefore not restricted to positive semidefinite
quadratic measurements.

\end{remark}

\begin{remark}
To make the phrase ``almost everywhere" precise, we specify a natural measure class on \(\cQ_d\). Every nonzero phase class determines uniquely its norm and the complex line it spans. Hence
\begin{equation}\label{eq:polar-quotient}
    \cQ_d\setminus\{[0]\}
    \cong (0,\infty)\times\CP^{d-1},
    \qquad
    [\bm{x}]
    \longmapsto
    \bigl(\norm{\bm{x}},[\bm{x}]_{\C^*}\bigr),
\end{equation}
where \([\bm{x}]_{\C^*}\) denotes the complex line spanned by \(\bm{x}\). We equip this product with the measure class induced by any smooth positive density and declare \(\{[0]\}\) to have measure zero. This definition is independent of the chosen density: locally, the ratio of any two smooth positive densities is a positive continuous function, and hence is bounded above and below by positive constants on relatively compact coordinate charts. Since the manifold admits a countable cover by such charts, all smooth positive densities have the same null sets. This convention agrees with that used in \cite[Section~2]{HuangRongWangXu2021}.
\end{remark}

\section{Preliminaries and Auxiliary Results on Maps Between Manifolds}
\label{sec:lem}
This section introduces several definitions and auxiliary results on smooth
manifolds and maps between manifolds that will be used in the proof of our
main theorem.

We  use the following standard fact from real-analytic geometry. If
\(\Omega\subset\mathbb{R}^n\) is a connected open set and
\(f:\Omega\to\mathbb{R}\) is a real-analytic function that is not identically
zero, then
\[
    \bigl\{u\in\Omega:f(u)=0\bigr\}
\]
has \(n\)-dimensional Lebesgue measure zero. A short proof of this fact is
given by Mityagin~\cite{Mityagin2020}.

We first specify what is meant by a null set on a manifold. Let \(M\) be an
\(n\)-dimensional smooth manifold, where
$
    n:=\dim_{\mathbb{R}}M.
$
A subset \(E\subset M\) is said to have \emph{measure zero in \(M\)} if, for
every smooth coordinate chart
$
    \phi:U\longrightarrow\phi(U)\subset\mathbb{R}^n,
$
the coordinate image
$
    \phi(E\cap U)
$
has \(n\)-dimensional Lebesgue measure zero in \(\mathbb{R}^n\). Here
\(U\subset M\) is an open set and \(\phi\) is a diffeomorphism from \(U\)
onto the open subset \(\phi(U)\) of \(\mathbb{R}^n\).
This definition is independent of the choice of coordinates. 

The following lemma is useful in our argument. 

\begin{lemma}\cite[Chapter~3, \S1, Exercise~4(a), p.~74]{Hirsch1976}
\label{lem:hirsch-critical}
Let \(X\) be a connected real-analytic manifold and let \(Y\) be a
real-analytic manifold. Suppose that
\[
    F:X\longrightarrow Y
\]
is a real-analytic map. Define the critical set of \(F\) by
\[
    \operatorname{Crit}(F)
    :=
    \left\{
        x\in X:
        \operatorname{rank}(DF_x)<\dim_{\mathbb R}Y
    \right\},
\]
where
$
    DF_x:T_xX\longrightarrow T_{F(x)}Y
$
denotes the differential of \(F\) at \(x\), and
$
    \operatorname{rank}(DF_x)
    :=
    \dim_{\mathbb R}\operatorname{Im}(DF_x).
$
If
$
    \operatorname{Crit}(F)\neq X,
$
then
$
    F^{-1}
    \bigl(
        F(\operatorname{Crit}(F))
    \bigr)
$
has measure zero in \(X\).
Here,
$
    F^{-1}(F(\operatorname{Crit}(F)))
    :=
    \{x\in X:F(x)\in F(\operatorname{Crit}(F))\}.
$

\end{lemma}

\begin{lemma}
\label{lem:max-rank-locus}

Let \(M\) be a connected real-analytic manifold of finite real
dimension
\(
    n:=\dim_{\mathbb R}M ,
\)
and let
\(
    N\geq1
\)
be an integer. Let
\(
    q:M\longrightarrow\mathbb R^N
\)
be a real-analytic map.

For every \(x\in M\), let
\[
    Dq_x:T_xM\longrightarrow
    T_{q(x)}\mathbb R^N\cong\mathbb R^N
\]
denote the differential of \(q\) at \(x\). Define
\(
    \operatorname{rank}(Dq_x)
    :=
    \dim_{\mathbb R}\operatorname{Im}(Dq_x).
\)
Let
\(
    r
    :=
    \max_{x\in M}
    \operatorname{rank}(Dq_x).
\)
Then
$
    0\leq r\leq \min\{n,N\}.
$
Assume that
\(
    r\geq1.
\)
Define the lower-rank locus
\[
    \Sigma_r
    :=
    \left\{
        x\in M:
        \operatorname{rank}(Dq_x)<r
    \right\}.
\]
Then \(\Sigma_r\) has measure zero in \(M\). 
\end{lemma}
\begin{proof}

Since \(r\) is the maximal rank of the differential of \(q\), by the
definition of the maximum there exists a point
$
  x_0 \in M
$
such that
$
    \operatorname{rank}(Dq_{x_0})=r .
$

Define the image space
\[
    V_0
    :=
    \operatorname{Im}(Dq_{x_0})
    \subseteq\mathbb R^N .
\]
Because
$
    \dim_{\mathbb R}V_0
    =
    \operatorname{rank}(Dq_{x_0}),
$
we have
$
    \dim_{\mathbb R}V_0=r .
$

Choose a linear map
$
    P:\mathbb R^N\longrightarrow\mathbb R^r
$
such that the restriction
$
    P|_{V_0}:V_0\longrightarrow\mathbb R^r
$
is a linear isomorphism.
Such a map exists by elementary linear algebra: choose a basis
$
    v_1,\ldots,v_r
$
of \(V_0\), extend it to a basis of \(\mathbb R^N\), and define
\(P\) by mapping \(v_j\) to the \(j\)-th standard basis vector of
\(\mathbb R^r\) and all remaining basis vectors to zero.

Now define
\[
    F:=P\circ q:
    M\longrightarrow\mathbb R^r .
\]
Because \(q\) is real analytic and \(P\) is linear, the map \(F\) is
real analytic.
For every \(x\in M\), the chain rule gives
$
    DF_x
    =
    P\circ Dq_x .
$
Indeed, the differential of the linear map \(P\) is \(P\) itself,
and hence
\[
    D(P\circ q)_x
    =
    DP_{q(x)}\circ Dq_x
    =
    P\circ Dq_x .
\]

At the point \(x_0\), we obtain
\[
\begin{aligned}
    \operatorname{Im}(DF_{x_0})
    &=
    \operatorname{Im}
    (P\circ Dq_{x_0})                                      \\
    &=
    P(\operatorname{Im}(Dq_{x_0}))                         \\
    &=
    P(V_0).
\end{aligned}
\]
Since
$
    P|_{V_0}:V_0\rightarrow\mathbb R^r
$
is an isomorphism,
$
    P(V_0)=\mathbb R^r .
$
Therefore,
$
    \operatorname{rank}(DF_{x_0})=r.
$

Define the critical set of \(F\) by
\[
    \operatorname{Crit}(F)
    :=
    \left\{
        x\in M:
        \operatorname{rank}(DF_x)<r
    \right\}.
\]
Since
$
    \operatorname{rank}(DF_{x_0})=r,
$
we have
$
    x_0\notin\operatorname{Crit}(F).
$
Consequently,
$
    \operatorname{Crit}(F)\neq M .
$
Therefore the assumptions of Lemma~\ref{lem:hirsch-critical} are satisfied. Hence
\(
    F^{-1}
    \bigl(
        F(\operatorname{Crit}(F))
    \bigr)
\)
has measure zero in \(M\).

We claim 
$
    \Sigma_r\subseteq\operatorname{Crit}(F),
$
which implies  
\[
    \Sigma_r
    \subseteq
    \operatorname{Crit}(F)
    \subseteq
    F^{-1}
    \bigl(
        F(\operatorname{Crit}(F))
    \bigr).
\]
We arrive at the conclusion. 

It remains to prove that
$
    \Sigma_r\subseteq\operatorname{Crit}(F).
$
Let
$
    x\in\Sigma_r .
$
Then, by definition,
$
    \operatorname{rank}(Dq_x)<r .
$
Since
$
    DF_x=P\circ Dq_x,
$
the image of \(DF_x\) satisfies
\[
    \operatorname{Im}(DF_x)
    =
    P(\operatorname{Im}(Dq_x)).
\]
Therefore,
$
    \operatorname{rank}(DF_x)
    \leq
    \operatorname{rank}(Dq_x).
$
Hence,
$
    \operatorname{rank}(DF_x)
    <
    r .
$
By the definition of the critical set,
$
    x\in\operatorname{Crit}(F).
$
Thus,
$
    \Sigma_r\subseteq\operatorname{Crit}(F).
$

\end{proof}

We next recall the mod-\(2\) degree theorem for smooth maps between
manifolds of the same dimension.
We first introduce the definition of the regular value of a map.

\begin{definition}\label{de:reg}
Let \(X\) and \(Y\) be smooth manifolds without boundary, both of real
dimension \(n\). Assume that \(X\) is compact and that \(Y\) is connected,
and let
\[
    f:X\longrightarrow Y
\]
be a smooth map.
A point \(y\in Y\) is called a \emph{regular value} of \(f\) if, for every
\(x\in f^{-1}(\{y\})\), the differential
\[
    Df_x:T_xX\longrightarrow T_yY
\]
is surjective. In particular, every point \(y\in Y\setminus f(X)\) is a
regular value because \(f^{-1}(\{y\})\) is empty.
\end{definition}

\begin{theorem}( \cite[\S4, pp.~20--25]{Milnor1997}).
\label{thm:mod-two-degree}
Let \(X\), \(Y\), and \(f:X\to Y\) be as in Definition~\ref{de:reg}.
If \(y,z\in Y\) are regular values of \(f\), then their fibers are finite
and satisfy
\[
    \# f^{-1}(\{y\})
    \equiv
    \# f^{-1}(\{z\})
    \pmod 2,
\]
where \(\#E\) denotes the cardinality of a finite set \(E\).
\end{theorem}

\begin{theorem}(\cite[Theorem~4.12, pp.~81--82]{Lee2013})\label{th:cons}
Suppose that \(F:X\to Y\) is smooth and that
\[
    \operatorname{rank}(DF_z)=s
\]
for every \(z\) in a neighborhood of a point \(x\in X\).
Then there exist smooth coordinate charts centered at \(x\) and
\(F(x)\) in which \(F\) has the local normal form
\[
    (u_1,\ldots,u_n)
    \longmapsto
    (u_1,\ldots,u_s,0,\ldots,0).
\]
\end{theorem}

The following Theorem is a direct consequence of
Theorem~\ref{thm:mod-two-degree}, Theorem~\ref{th:cons}, and
Lemma~\ref{lem:max-rank-locus}. 

\begin{theorem}
\label{thm:singleton-fiber}
Let \(M\) be a compact, connected real-analytic manifold without boundary,
with real dimension
\(
    n:=\dim_{\mathbb{R}}M\geq1.
\)
Let
\(
    q:M\longrightarrow\mathbb{R}^n
\)
be a real-analytic map. Set
\[
    U_q
    :=
    \left\{
        x\in M:
        q^{-1}\bigl(\{q(x)\}\bigr)=\{x\}
    \right\}.
\]
Then \(U_q\) has measure zero in \(M\).

\end{theorem}

\begin{proof}
 For each \(x\in M\), let
\[
    Dq_x:T_xM\longrightarrow
    T_{q(x)}\mathbb{R}^n\cong\mathbb{R}^n
\]
denote the differential of \(q\) at \(x\).
Set
\(
    r:=\max_{x\in M}\operatorname{rank}(Dq_x).
\)

If \(r=0\), then \(Dq\equiv0\), and hence \(q\) is  constant. Since \(\dim_{\mathbb R}M\geq1\),
no fiber is a singleton. Here we use \(M\) is connected.

Suppose that \(1\leq r<n\). By
Lemma~\ref{lem:max-rank-locus}, the set
\[
    \Sigma_r
    =
    \{x\in M:\operatorname{rank}(Dq_x)<r\}
\]
is null. If \(x\notin\Sigma_r\), then an \(r\times r\) Jacobian minor is
nonzero at \(x\), and hence remains nonzero near \(x\). By the maximality
of \(r\), the rank is therefore identically \(r\) near \(x\). According to
Theorem \ref{th:cons}, the local fiber through \(x\) has
dimension
\(
    n-r\geq1.
\)
Thus \(x\notin U_q\), so \(U_q\subseteq\Sigma_r\). According to Lemma~\ref{lem:max-rank-locus}, $U_q$ has measure zero.

Finally, suppose that \(r=n\). Again, the critical set
\[
    \Sigma_n
    =
    \{x\in M:\operatorname{rank}(Dq_x)<n\}
\]
is null. 
For the aim of contradiction, we assume that $U_q\not\subseteq \Sigma_n$.
This implies that there exists a point \(x\in U_q\setminus\Sigma_n\).
Since \(x\notin\Sigma_n\), the point \(q(x)\) is a regular value of \(q\).
Moreover, because \(x\in U_q\), by the definition of \(U_q\) we have
\(
    \#q^{-1}(\{q(x)\})=1.
\)
Since \(M\) is compact, \(q(M)\) is compact, so one may choose
\(y_0\in\mathbb R^n\setminus q(M)\). Then \(y_0\) is also a regular value,
with
\(
    \#q^{-1}(\{y_0\})=0.
\)
Theorem \ref{thm:mod-two-degree}
 implies that the cardinalities of
the fibers over any two regular values have the same parity, giving
\[
    1\equiv0\pmod2,
\]
a contradiction. Hence \(U_q\subseteq\Sigma_n\). Following  Lemma~\ref{lem:max-rank-locus}, $U_q$ has measure zero.

\end{proof}

\section{Proof of \Cref{thm:main}}

In this section, we apply the auxiliary results established in
Section~\ref{sec:lem} to prove the main theorem,
Theorem~\ref{thm:main}.

\begin{proof}[Proof of \Cref{thm:main}]
It is enough to consider the case where
\(
    m:=2d-1.
\)
Recall that
our aim is to show that
\[
    \mathcal U_F
    :=\left\{[\bm{x}]\in\cQ_d:
       \cM_F^{-1}\bigl(\cM_F([\bm{x}])\bigr)=\{[\bm{x}]\}
      \right\}
\]
has measure zero in $\cQ_d$ where $F=(\bm{f}_j)_{j=1}^m\subset \C^d$.

A simple argument shows that the conclusion holds if  \(F\) does not span \(\C^d\).
We now assume that \(F\) spans \(\C^d\). Define
\[
    S:=\sum_{j=1}^{m}\bm{f}_j\bm{f}_j^*.
\]
Then \(S\) is positive definite. 
Define
\(
    \tcM_F:\CP^{d-1}\longrightarrow\mathbb R^m
\)
by
\[
    \tcM_F([\bm{u}]_{\C^*})
    :=
    \left(
        \frac{|\bm{f}_j^*\bm{u}|^2}
             {\bm{u}^*S\bm{u}}
    \right)_{j=1}^{m},
    \qquad \bm{u}\neq\bm{0}.
\]
This map is well defined because both the numerator and denominator are
multiplied by \(|c|^2\) when \(\bm{u}\) is replaced by
\(c\bm{u}\), where \(c\in\C\setminus\{0\}\). It is real analytic in
the standard projective coordinate charts. Moreover,
\[
    \sum_{j=1}^{m}
    \frac{|\bm{f}_j^*\bm{u}|^2}
         {\bm{u}^*S\bm{u}}
    =1,
\]
so the image of \(\tcM_F\) is contained in the affine hyperplane
\[
    \cH_m
    :=
    \left\{
        \bm{t}\in\mathbb R^m:
        \sum_{j=1}^{m}t_j=1
    \right\}.
\]
Since
\[
    \dim_{\mathbb R}\cH_m
    =
    m-1
    =
    2d-2
    =
    \dim_{\mathbb R}\CP^{d-1},
\]
we fix an affine identification
$
    \cH_m\cong\mathbb R^{2d-2}.
$
With this identification understood, we henceforth regard \(\tcM_F\) directly as a map
\[
    \tcM_F:
    \CP^{d-1}\longrightarrow\mathbb R^{2d-2},
\]
and retain the same notation \(\tcM_F\) for this map.

The manifold \(\CP^{d-1}\) is compact, connected, real analytic, and
without boundary. Therefore, \Cref{thm:singleton-fiber} implies that
\[
    B_F
    :=
    \left\{
        \bell\in\CP^{d-1}:
        \tcM_F^{-1}(\tcM_F(\bell))=\{\bell\}
    \right\}
\]
has measure zero in \(\CP^{d-1}\).
Under the polar decomposition
\[
    \cQ_d\setminus\{[0]\}
    \cong
    (0,\infty)\times\CP^{d-1},
\]
write a nonzero phase class as \((r,\bell)\in (0,\infty)\times\CP^{d-1}\), represented by
\(\bm{x}=r\bm{u}_\bell\), where \(\|\bm{u}_\bell\|=1\) and
\([\bm{u}_\bell]_{\C^*}=\bell\). 
By a slight abuse of notation, we use the same symbol \(\cM_F\) for the restriction of \(\cM_F\) to \(\cQ_d\setminus{[0]}\), expressed in these polar coordinates, and write \(\cM_F(r,\bell):=\cM_F([r\bm{u}_\bell])\). This definition is independent of the choice of the unit representative \(\bm{u}_\bell\).

We claim that 
\begin{center}
the phase class \((r,\bell)\in (0,\infty)\times\CP^{d-1}\) has a singleton fiber under
\(\cM_F\) if and only if \(\bell\in B_F\). 
\end{center}
Hence
\(
    \mathcal U_F
    =
    \{[0]\}
    \cup
    \bigl((0,\infty)\times B_F\bigr).
\)
Since \(B_F\) is null in \(\CP^{d-1}\), Tonelli's theorem shows that
\(
    (0,\infty)\times B_F
\)
is null.
By the measure convention on \(\cQ_d\), the singleton \(\{[0]\}\) is
also null. Therefore,
\(
    \mathcal U_F
\)
has measure zero in \(\cQ_d\), as claimed.

It remains to prove the claim that the phase class \((r,\bell)\in(0,\infty)\times\CP^{d-1}\) has a singleton fiber under \(\cM_F\) if and only if \(\bell\in B_F\).
 For
\(\bell\in\CP^{d-1}\), choose a unit vector \(\bm{u}_\bell\) spanning
\(\bell\), and define
\[
    a_F(\bell)
    :=
    \bigl(|\bm{f}_j^*\bm{u}_\bell|^2\bigr)_{j=1}^{m},
    \qquad
    \sigma_F(\bell)
    :=
    \sum_{j=1}^{m}|\bm{f}_j^*\bm{u}_\bell|^2
    =
    \bm{u}_\bell^*S\bm{u}_\bell.
\]
These quantities are independent of the choice of the unit representative
of \(\bell\), and
\[
    \sigma_F(\bell)>0,
    \qquad
    \tcM_F(\bell)
    =
    \frac{a_F(\bell)}{\sigma_F(\bell)}.
\]
A simple observation is that
\(
    \cM_F(r,\bell)
    =
    r^2a_F(\bell).
\)

For $(r,\bell), (r',\bell') \in(0,\infty)\times\CP^{d-1}$,
suppose that
\begin{equation}\label{eq:cMFF}
  \cM_F(r,\bell)= \cM_F(r',\bell'),
\end{equation}
i.e.,$ r^2a_F(\bell)
    =
    r'^2a_F(\bell')$.
Summing the coordinates gives
\(
    r^2\sigma_F(\bell)
    =
    r'^2\sigma_F(\bell').
\)
Under the assumption in \eqref{eq:cMFF}, it also follows that
 $(r,\bell)\neq (r',\bell')$ if and only if $\bell\neq \bell'$.
Dividing the measurement vectors by this common positive number yields
\[
    \tcM_F(\bell)=\tcM_F(\bell').
\]
Conversely, if \(\tcM_F(\bell)=\tcM_F(\bell')\), then choosing
\(
    r'
    :=
    r\sqrt{\frac{\sigma_F(\bell)}
                  {\sigma_F(\bell')}}
\)
gives
\(
    r'^2a_F(\bell')
    =
    r^2a_F(\bell),
\) i.e., $
  \cM_F(r,\bell)= \cM_F(r',\bell')
$.
Consequently, the phase class \((r,\bell)\) has a singleton fiber under
\(\cM_F\) if and only if \(\bell\in B_F\). Hence
\[
    \mathcal U_F
    =
    \{[0]\}
    \cup
    \bigl((0,\infty)\times B_F\bigr).
\]
Here \([0]\) is a singleton fiber because \(S\) is positive definite.

\end{proof}

\section{Conclusion}

In this paper, we establish that the minimum number of measurements required
for almost-everywhere phase retrieval (PR-ae) in \(\mathbb{C}^d\) is exactly
\(2d\). The main ingredient of our proof is a differential-topological
approach based on smooth maps between manifolds. In contrast to previous
approaches relying primarily on algebraic geometry and the study of
algebraic varieties \cite{HuangRongWangXu2021}, our method exploits the geometric properties of smooth
maps, such as regular values, critical points, and the structure of
singleton fibers. This provides a different perspective for analyzing
almost-everywhere injectivity problems.

We believe that the same manifold-based framework may also be useful for
studying almost-everywhere low-rank matrix recovery problems. For example,
in \cite{RongWangXu2021}, one of the fundamental questions is to characterize
the minimum number of linear measurements required to guarantee that a
generic rank-\(r\) matrix can be uniquely recovered. 
The smooth manifold structure of fixed-rank matrix varieties suggests that
the techniques developed in this paper may provide a new approach toward
such problems.

\vspace{0.5cm}

\textbf{Statement on the Use of AI:}

During the preparation of this manuscript, the authors used ChatGPT (OpenAI) as a supporting tool for mathematical exploration, language refinement, and improvement of exposition. All AI-assisted suggestions and generated text were carefully examined, mathematically verified, and substantially revised or rewritten by the author. The authors take full responsibility for the accuracy, validity, and final form of the manuscript.

\end{document}